\documentclass[runningheads]{llncs}

\usepackage[T1]{fontenc}

\makeatletter
\renewcommand{\@Opargbegintheorem}[4]{%
  #4\trivlist\item[\hskip\labelsep{#3#2\@thmcounterend}]}
\makeatother

\usepackage{graphicx}
\usepackage{amsmath}
\usepackage{amssymb}
\newcommand{\rev}[1]{{#1}}

\usepackage{setspace}
\usepackage{xspace}
\usepackage[capitalise]{cleveref}
\usepackage{listings}
\usepackage{booktabs}
\usepackage{varwidth}

\usepackage[scaled=.85]{courierten}

\usepackage{caption}
\usepackage{subcaption}

\newcommand{\sig}{SIGNAL\xspace}
\newcommand{\sigCC}{$\mathit{S\!C\!C}$\xspace}

\newcommand{\RAExtended}{$\text{RA}^{\mathcal{P}}\!$\xspace}

\begin{document}
\title{Expressive Power and Complexity Results for SIGNAL, an Industry-scale Process Query Language}

\titlerunning{Expressive Power and Complexity Results for SIGNAL}
%
\author{Timotheus Kampik\inst{1,2}\orcidID{0000-0002-6458-2252} \and
{Cem~Okulmus\inst{2}\orcidID{0000-0001-7116-9338}}}
\authorrunning{T. Kampik and C. Okulmus}
%
\institute{SAP, Berlin, Germany 
 \and Department of Computing Science, {Ume{\aa}} University,  \\ {Ume{\aa}}, Sweden\\
\email{\{tkampik,okulmus\}@cs.umu.se} 
}

\maketitle

\begin{abstract}
With the increased adoption of process mining, there is also a need for practical solutions that work at industry scales. In this context, process querying methods (PQMs) have emerged as an important tool for drawing inferences from event logs.
Here, it can be expected that industry approaches differ from academic ones, due to practical engineering and business considerations.
To understand what is at the core of industry-scale PQMs, a formal analysis of the underlying languages can provide a solid foundation.
To this end, we formally analyse SIGNAL, an industry-scale language for querying business process event logs developed by a large enterprise software vendor.
The formal analysis shows that the core capabilities of SIGNAL, which we refer to as the \emph{SIGNAL Conjunctive Core}, are more expressive than relational algebra and thus not captured by standard relational databases. We provide an upper-bound on the expressiveness via a reduction to semi-positive Datalog, which also leads to an upper bound of P-hard for the data complexity of evaluating SIGNAL Conjunctive Core queries.
The findings provide first insights into how (real-world) process query languages are fundamentally different from the more generally prevalent structured query languages for querying relational databases and provide a rigorous foundation for extending the existing capabilities of the industry-scale state-of-the-art of process data querying.
\keywords{Process mining \and Process querying \and Databases}
\end{abstract}

\section{Introduction}
\label{sec:introduction}
The increased industry adoption of process mining requires practical solutions that work at industry scales. Process querying methods (PQMs)  have emerged as central, laying the foundations of industry-scale tools for drawing inferences from the event logs that collect the process traces as recorded by enterprise systems~\cite{kampik2023signal,DBLP:books/sp/22/P2022,DBLP:books/sp/22/0001ABSGK22}. 
Intuitively, these methods provide structured query languages that focus on querying representations common in business process management, such as \emph{graphs} on modalities that are traditionally considered important in the domain, such as \emph{logical time}.
Despite the importance of PQMs in academia and practice, so far and to the best of our knowledge, none of these industry-scale methods have been subject to a detailed formal analysis. This is understandable, since these tools have been developed using a top-down approach, with the need to find practical solution quickly to meet industry demand for process mining insights.  
However, the lack of a formal understanding of these industry-scale PQMs, we argue, leads to an interesting and relevant research challenge.
First, establishing such an understanding firmly answers the question as to how PQMs can theoretically scale with the size of the underlying data, thus showing whether these tools are indeed an efficient and practical solution beyond individual examples.
Second, a complete formal analysis opens the way to more expressive extensions of existing PQMs that are still ``well-behaved'' in terms of computational complexity.
Finally, and most importantly, a formalisation of a PQM allows us to compare its expressive power to other data management approaches and see if PQM methods are indeed adding something new compared to the state of the art in database management, such as relational or NoSQL database systems.
We thus see this work as the first step towards the goal of a formal description of, ideally, all PQMs that are publicly documented and can thus be studied.  There are already numerous  PQMs described in the literature~\cite{DBLP:books/sp/22/P2022}, such as PQL~\cite{POLYVYANYY2024102337}, Celonis PQL~\cite{DBLP:books/sp/22/0001ABSGK22}, and SIGNAL~\cite{kampik2023signal}, as well as several graph-based languages for process querying that have gained---while nascent---substantial research attention in recent years~\cite{DBLP:conf/icpm/Jalali20,DBLP:conf/er/KhayatbashiHJ23,DBLP:journals/jodsn/EsserF21,DBLP:journals/spe/BeheshtiBM18}.

For this paper we focus our attention on SIGNAL, which is actively developed by a large enterprise software vendor in a process mining product\footnote{Here, we can rely on ``openly'' available documentation, cf. \url{help.sap.com/docs/signavio-process-intelligence/signal-guide/syntax} (last accessed at 2024-02-20) -- as well as on a technical report~\cite{kampik2023signal} shared by the vendor.}. We note that our focus on SIGNAL should be understood as providing the necessary starting point in our endeavour to better understand the formal properties of PQMs, and not as the end goal. Still, we argue that focusing on an industry-scale language \emph{is} desirable to start with, as it helps unearth practical considerations emerging from real-world engineering and business requirements.

Thus, our research questions for this work are as follows:
\begin{enumerate}
    \item Which well-known query languages can and cannot match the process querying language SIGNAL with respect to expressive power?
    \item What are the complexity bounds of the process querying language SIGNAL?
\end{enumerate}
Answering Question 1. tells us whether state-of-the-art process data querying requires fundamentally different capabilities than, e.g., querying a traditional relational database using an SQL dialect; answering Question 2. tells us to what extent it is (theoretically) possible to query large amounts of data fast.
\newpage
\paragraph{Our Contributions.}  In this work we present the following results: 
\begin{itemize}
    \item We provide the first formalisation of the conjunctive core of the SIGNAL query language. \rev{This formalisation is based on publicly available documentation and grammar specification; we also had access to a running SIGNAL instance that we could query via a RESTful API.}
    \item We show that the expressive power of the conjunctive core of SIGNAL cannot be captured by relational algebra.
    \item We provide an upper-bound on the data complexity of evaluating the conjunctive core of SIGNAL by a translation to semi-positive Datalog. It is left as an open question if this also forms the lower-bound, which would lead to a P-completeness result.
\end{itemize}
The rest of the paper is structured as follows.
We provide the formal preliminaries of event logs and cases, presented in the context of the relational model, in Section~\ref{sec:preliminaries}.
Based on the preliminaries, we provide a formalisation of the core of SIGNAL in Section~\ref{sec:the_core_of_the_process_mining_query_language_sig}.
We then show that the SIGNAL core cannot be expressed in relational algebra and translate the language to semi-positive Datalog in order to allow for a straightforward theoretical analysis in Section~\ref{sec:bounds}.
The results of our analysis show that the core of SIGNAL is more expressive than relational algebra and not more expressive than semi-positive Datalog; hence, we can also establish that the upper bound in the data complexity is in polynomial time. The paper concludes in Section~\ref{sec:conclusion} with a call for future work that will ideally expand the expressive power and complexity results to other PQMs.

\section{Preliminaries}
\label{sec:preliminaries}
We proceed to give a succinct summary of the relational model. For a more detailed introduction of this and related topics on the foundations of databases, we refer to~\cite{DBLP:books/aw/AbiteboulHV95}. A domain (or attribute) $D$ is a set of values, such as the natural numbers, alphabetic words of various length, and so on. A tuple $t \subseteq D_1 \times \dots \times D_n$ is an $n$-ary combination of values from $n$ many domains. A relation $R \subseteq 2^{D_1 \times \dots \times D_n}$ is a set of tuples. We distinguish between the \emph{instance} of a relation, which is the set described above, and the \emph{schema} of all relations of this type via the set of domains (or attributes) that make up the values in the tuples in the relation instance. When writing about the schema of a relation, we shall equate it simply with the list of its domains. We assume that any reader is familiar with the relational algebra operations, like projection ($\pi$), selection ($\sigma$), various joins, and set operations. For readers that are not familiar with these, we again refer to the excellent textbook from Abiteboul et al.~\cite{DBLP:books/aw/AbiteboulHV95}.
Another concept we need is the notion of conjunctive query (CQ). Formally, a CQ has a body,  made up of relational atoms, and a head, a subset of variables occurring in the body, also called the answer variables:
\[
   q(\mathbf{x}) \leftarrow R_1(\mathbf{y_1}), \dots , R_n(\mathbf{y_n}),
 \]

where we have that $\mathbf{x} \subseteq \bigcup_{1 \leq i \leq n} \mathbf{y_i}$. CQs correspond to SELECT-FROM-WHERE queries in SQL where we only allow equations between attributes in the WHERE clause. Another correspondence is to the positive fragment of first-order logic. Despite its structural simplicity, the problem of answering a CQ---defined as the problem  of finding all mappings between the variables of the body atoms to instances of the respective relations in a given database---is known to be NP-complete. This combination of structural simplicity while still retaining a high complexity for the basic problem of query answering has made CQs an important object of study in database theory. 

\paragraph{Computational Complexity of Query Answering.} For the problem of query answering as defined above, there are three notions of computational complexity. The first is to consider both the query and the database instance to be part of the input -- this gives us the \emph{combined complexity} of query answering.  As discussed, already fairly simple queries, such as CQs, can be hard to answer. However, usually the query is fairly small and thus one might want to investigate how the problem of query evaluation scales with the database  size. In this case we treat the query as a constant. This gives rise to the notion of \emph{data complexity}. And finally, if one only wants to study how the problem scales with the size of the query, one can set the database size to constant and this gives the notion of\emph{ query complexity}.

Due to the prevalence of temporal data in the form of event logs in process mining, we introduce a special timestamp attribute $\mathcal{T}$\!.
For the sake of simplicity, we assume that a timestamp here refers to a natural number (or $0$), representing the UNIX timestamp.

We present here a simplified view on the concepts of event logs and cases, as they play a crucial role in process mining and to understand the query language SIGNAL we wish to present. While there are different notations and terminologies for these concepts in the literature~\cite{DBLP:books/sp/Aalst16}, we will present these concepts in the context of the (flat) relational model, which most database theorists and users will likely be familiar with. 

\begin{table}[t]
\centering
\caption{An event log without case attributes.}
\label{tab:cases}
\begin{tabular}{ ccccc }
\toprule
\textbf{event\_ID} & \textbf{case\_ID} & \textbf{timestamp} & \textbf{event\_name}  & \textbf{status} \\
\midrule
 e0001 & 0001 & 1675086864052 & \texttt{Review request} & \texttt{NEW} \\ 
\midrule
 e0002 & 0002 & 1675147138009 & \texttt{Review request} & \texttt{NEW} \\ 
\midrule
 e0003 & 0001 & 1675160180724 & \texttt{Calculate terms} & \texttt{WIP} \\ 
\midrule
 e0004 & 0002 & 1675213914098 & \texttt{Define terms} & \texttt{WIP} \\ 
\midrule
 e0005 & 0001 & 1675220315296 & \texttt{Prepare contract} & \texttt{WIP} \\ 
\midrule
 e0006 & 0002 & 1675282027657 & \texttt{Prepare contract} & \texttt{WIP} \\ 
\midrule
 e0007 & 0002 & 1675414104525 & \texttt{Send quote} & \texttt{SENT} \\ 
 \bottomrule
\end{tabular}
\end{table}

\begin{definition}[Events and Event Logs]\label{def:Events}
    An \emph{event log} $L$ is a relation with the schema $(\mathit{Eid}, \mathit{Cid}, \mathcal{T}\!, A_1, \dots A_n)$, where the first three attributes are fixed, 
    \rev{namely $\mathit{Eid}$ for the \emph{event id}, $\mathit{Cid}$ for the case id and $\mathcal{T}$ to indicate time}, and we allow an $n$-ary set of attributes for the rest of the event log schema, which we shall call the \emph{event attributes}. We will refer to a single tuple inside $L$ as an \emph{event}. We furthermore define the functional dependencies that hold for any event log $L$: both $(\mathit{Eid},\mathit{Cid})$ and $(\mathit{Cid},\mathcal{T})$ shall be key candidates. Formally, for no two events $e, e' \in L$ where $e \not= e'$, $e = (\mathit{eid}, \mathit{cid}, t, a_1, \dots a_n) $ and $e' = (\mathit{eid}',\mathit{cid}', t', a'_1, \dots a'_n )$ such that $\mathit{cid} = \mathit{cid}'$ it may hold that $t = t'$ or $\mathit{eid} = \mathit{eid}'$. Some further notation: we define the function $\mathit{att}$ to project an event to its event attributes, e.g.: $\mathit{att}(e) = (a_1, \dots, a_n)$. Due to the timestamp attribute $\mathcal{T}$ in events, we can also define a ordering over events \rev{of the same case}: $e \succeq e' $ (resp. $e \succ e'$) holds, if we have $t \geq t'$ (resp. $t > t'$).
\end{definition}

To easily refer to the set of cases in an event log $L$ with schema $(\mathit{Eid}, \mathit{Cid}, \mathcal{T}\!, A_1, \dots A_n)$, we introduce the notation $\mathit{Cases}(L) = \mathit{Cid}$. When we need to refer to values inside events, we will by slight abuse of notation use the attribute names as functions: for $e = (\mathit{eid}, \mathit{cid}, t, a_1, \dots a_n) $, we can then use $\mathcal{T}(e)$ to refer to $t$, $A_1(e)$ to $a_1$, and so on.
 
Let us give a simple example that illustrates how a set of events can be represented.

\begin{example}\label{ex:events}
Consider a quote creation process, e.g., for credits in retail banking.
First, the request is reviewed (\texttt{Review request}).
Then, standard terms are calculated, if applicable (\texttt{Calculate terms}).
Otherwise, custom terms are defined (\texttt{Define terms}).
In either case, the contract is subsequently prepared (\texttt{Prepare contract}) and finally sent out (\texttt{Send quote}).
The events generated during process execution have the following attributes, besides timestamp, case ID, and event ID:
    \emph{event\_name}, giving the event a human-interpretable meaning; and 
    \emph{status}, describing the status change that occurs (i.e., the resulting case status) when the event occurs.
We then have the schema $(\mathit{Eid}, \mathit{Cid}, \mathcal{T}\!, \text{\emph{event\_name}}, \text{\emph{status}})$.
We provide an example for such an event log in \cref{tab:cases}, with events ordered by end timestamp---the assumption that timestamps are unique gives a total order on the set of events.
This simplistic example is missing case attributes, as we will introduce another relation to take care of those. 
\end{example}

\begin{definition}[Cases and Their Events]\label{def:Cases}
    A \emph{case set} C is a relation  $(\mathit{Cid}, B_1, \dots, B_\ell)$, where $(\mathit{Cid})$ is the sole key candidate and $B_1, \dots, B_\ell$ are the \emph{case attributes}.  In other words, this means that no two distinct entries of $C$ may share the same case id; case ids are thus unique across the relation $C$.
    We refer to tuples $c \in C$ as \emph{cases}, and by slight  abuse of notation we shall identify $c$ via its value for $\mathit{Cid}$.
    For a given event log $L$ and case $c \in \mathit{Cid}$, we define its set of  events (or event set) $E_c \subseteq L$ as $\{ e \mid e \in L,  \mathit{Cid(e)} = c  \}$.
\end{definition}

Due to the functional dependency of $L$ on $(\mathit{Cid},\mathcal{T})$, we have that every $e' \in E_c$ has a unique timestamp value for $\mathcal{T}$ and thus $\succeq$ acts as a total order over the elements in $E_c$ for any $c \in \mathit{Cid}$.

\begin{example}\label{ex:cases}
Let us extend the event collection as specified in Example~\ref{ex:events} with case attributes.
Our case set is defined as $(\mathit{Cid}, \text{\emph{customer\_ID}}, \text{\emph{terms}})$.  Due to \cref{def:Cases}, it follows that the attribute \emph{customer\_ID} identifies the case's customer and the attribute \emph{terms} logs the terms that apply for the case (\texttt{standard} or \texttt{custom} terms).
In the event log from Example~\ref{ex:events}, we have two cases; case $0001$ has the \emph{customer\_ID} $C0001$ and \texttt{standard} terms, whereas case $0002$ has the the \emph{customer\_ID} $C0002$ and \texttt{custom} terms.
\end{example}

\paragraph{Nested structures in event logs.} In the literature, and in the implementation of SIGNAL, event logs are in fact nested relations.  
In order to simplify the presentation, we do not assume such a nested relational model for this paper. The first reason is that the (flat) relational model is conceptually easier to grasp. The second reason is that it allows us to compare the expressiveness of SIGNAL with the expressiveness of relational algebra, without any support for nested relations. Because SIGNAL does not allow arbitrary nested structures
this simplification seems justified, as we do not need the full power of nested relational algebras~\cite{DBLP:books/aw/AbiteboulHV95}.

\section{The Conjunctive Core of \sig}
\label{sec:the_core_of_the_process_mining_query_language_sig}

We begin by stating the formal syntax of a subset of \sig. We restrict ourselves to a  fragment of \sig which is already expressive but hopefully still easy to grasp.

\begin{definition}[\sig  Conjunctive Core] 
\label{def:SCC}
To simplify the presentation and focus on the core aspects of the \sig query language, we define a subset of \sig queries, that we shall call \emph{\sig Conjunctive Core}, or \sigCC for short. 
We define the syntax of \sigCC queries in Extended Backus–Naur Form\footnote{Recall that as a reference, the complete syntax for \sig, beyond \sigCC, can be found in the official guide provided at \url{help.sap.com/docs/signavio-process-intelligence/signal-guide/syntax} (last accessed at 2024-02-20).}. 
{ \normalfont \small \begin{center}
    {
        \begin{lstlisting}[language=SQL, 
        commentstyle=\itshape, 
        basicstyle=\ttfamily,escapeinside={(*}{*)},morekeywords={START, MATCHES,BEHAVIOUR}]
<scc>           :=  SELECT <varlist>
                      FROM <Eventlog name>
                    [WHERE <conditions>]
<conditions>    :=  <condition> AND <conditions> | <condition> 
<condition>     :=  <var> = <var>  | <var> = <const>
                    <var> MATCHES <pattern> |
                    BEHAVIOUR <behaviours> MATCHES <pattern> 
<varlist>       :=  <var> |  <var> , <varlist> 
<behaviourCond> :=  <var> = <var> | <behaviourCond> AND <behaviourCond>
<behaviours>    :=  <behaviourCond> AS <var> | <behaviours> , <behaviours> 
<pattern>       :=  <pattern> (*$\rightsquigarrow$*) <pattern> | <pattern> (*$\rightarrow$*) <pattern> | <id> |
                    <pattern>* | ANY | START <pattern> | <pattern> END
<id>            :=  string | <var> | <id> OR <id> | NOT ( <id> )  
        \end{lstlisting}

}
\end{center}}
We assume the set of variables (\texttt{<var>}) and constants (\texttt{<const>}) to be system defined and omit a formal definition.
\end{definition}

\begin{example} \label{ex:SIGCC}
The simplest \sigCC{} queries use only sets of equalities and constant assignments. 

{  \small       \begin{lstlisting}[language=SQL, 
        commentstyle=\itshape,
        basicstyle=\ttfamily]
        SELECT case_id, event_name, event_time
          FROM eventlog 
         WHERE case_id = 2 AND event_name = "package received"
        \end{lstlisting}
}
We can also define pattern matching clauses. 

{  \small       \begin{lstlisting}[language=SQL, 
        commentstyle=\itshape,
      basicstyle=\ttfamily,escapeinside={(*}{*)},morekeywords={MATCHES} ]
        SELECT case_id, event_name, event_time
          FROM eventlog 
         WHERE event_name MATCHES ('package_sent' (*$\rightsquigarrow$*) 'package_accepted')
        \end{lstlisting}
}

Lastly, we give an example for a behaviour pattern formula.

{  \small       \begin{lstlisting}[language=SQL, 
        commentstyle=\itshape,
      basicstyle=\ttfamily,escapeinside={(*}{*)},morekeywords={MATCHES,BEHAVIOUR} ]
    SELECT case_id, event_name, event_time
      FROM eventlog 
     WHERE BEHAVIOUR A AS event_name = 'Review request' AND status = NEW, 
                     B AS event_name = 'Send quote' AND status = SENT
            MATCHES (A (*$\rightsquigarrow$*) B)
        \end{lstlisting}
}
  
\end{example}

As the name of \sigCC{} suggests, we are focusing on all SIGNAL queries that correspond to the well known formalism of \emph{conjunctive queries} in database theory, with the only extension being the ability to capture \emph{patterns} over event sets. In the next section, we will also briefly consider more expressive fragments of \sig{}, such as those that permit to have nested queries (also known as subqueries). 

We proceed to define the set of patterns, beginning with simple patterns that can only refer to values of a single event attribute.

\begin{definition}[Simple Pattern]\label{def:pattern}
We are given an event log $L$.
A simple pattern over $L$ is a pair $\langle A_i, \mathcal{P}_s  \rangle$, where $A_i$ is an event attribute in $L$ and $\mathcal{P}$ is a \emph{simple pattern formula}. 
We first define the notion of \emph{event identifiers}, which may be used in simple pattern formulas. We shall define the set of event identifiers inductively.

\begin{itemize}
    \item Every value $a \in A_i$ is an event identifier. 
    \item If $e$ and $e'$ are event identifiers, then so are: $e \lor e'$, $\mathit{not}( e' )$. 
\end{itemize}
Now we can define the set of simple event patterns.
\begin{itemize}
    \item $\mathit{any}$ is a simple pattern formula.
    \item Every event identifier is a simple pattern formula. 
    \item If $P'$ and $Q$ are simple pattern formulas, then so are: \\ $P' \rightarrow Q$, $P' \rightsquigarrow Q $, ${P'}^*, \mathit{start}(P'),(P')\mathit{end}$. 
\end{itemize}    
\end{definition}
Let us highlight that similar patterns are common in process mining approaches utilising regular languages, such as implementations of DECLARE~\cite{DBLP:conf/caise/CiccioBCM15} and graph-based PQMs (regular path queries)~\cite{DBLP:journals/spe/BeheshtiBM18,DBLP:journals/jodsn/EsserF21}.

\rev{
\begin{example} 
We shall use the schema we presented in~\Cref{ex:events} to give an example of a simple pattern formula. Let us first fix the following constants.\\
$    n_0 = \text{\tt 'Review Request'}  \text{\phantom{word}}  n_1 = \text{\tt'Calculate terms'} \text{\tt\phantom{word}}  n_2 = \text{\tt'Send quote'} $
An example of a simple pattern formula is: 
$p = ( n_0 \rightarrow ((n_1 \lor n_0)
 \rightsquigarrow n_2 )^* \mathit{end})$. 
\end{example}
}

\rev{Simple patterns cannot address conditions over multiple event attributes.}
To overcome this limitation, \sigCC{} also provides patterns that can refer to behaviours over multiple event attributes.  We call these behaviour patterns.

\begin{definition}[Behaviour Pattern]\label{def:complexPattern}
    We are given an event log $L$, with schema $(\mathit{Cid},\mathit{Eid}, \mathcal{T}, A_1, \dots, A_n)$.
    A behaviour pattern over $L$ is a pair $\langle  B, \mathcal{P}_b  \rangle$, where $B$ is a \emph{behaviour matching set} and $\mathcal{P}_b$ is a \emph{behavioural pattern formula}. A behaviour matching is a function $\sigma_x : A_1 \times \dots \times A_n \rightarrow \{ \top, \bot  \} $ from the event attributes to either $\top$ or $\bot$. We can also identify each such function with an identifier $x$. We first define the set of behaviour identifiers: 
    \begin{itemize}
        \item Every value $x$ for some behaviour matching $\sigma_x$, is a behaviour identifier
        \item If $b$ and $b'$ are behaviour identifiers, then so are $b \lor b'$, $\mathit{not}(b)$.
    \end{itemize}
    With this, we can now present the definition of behavioural pattern formulas. 
    \begin{itemize}
        \item $\mathit{any}$ is a behavioural pattern formula.
        \item Every behavioural identifier is a behavioural pattern formula. 
        \item If $P'$ and $Q$ are behavioural pattern formulas,  then so are: \\  $P' \rightarrow Q, P' \rightsquigarrow Q, {P'}^*, \mathit{start}({P'}),(P')\mathit{end}$. 
    \end{itemize} 
\end{definition}
\rev{\begin{example}
To give some intuition for behaviour pattern formulas, we focus on the behaviour matching function, which is the only difference to simple pattern formulas. We fix the schema to the one in \Cref{ex:events}. 
\[
    b_1 = ( \text{\bf status} = \text{\tt NEW} \lor  \text{\bf event\_name} = \text{\tt Review Request}   ) 
\] In $b_1$ we see one natural example of how to fix such functions: we define a logical formula and evaluate it over the tuple in the event log. Any tuple that satisfies it is matched to $\top$, otherwise to $\bot$. We note, however, that \Cref{def:complexPattern} does not impose a specific formalism on behaviour matching functions, as long as  tuples in event logs can be accepted (returning $\top$) or rejected (returning $\bot$).
\end{example} 
}

Note that behaviour patterns are strictly more general than simple patterns, and indeed one could define simple patterns as a special case of behavioural patterns, where the behaviour matching function is only considering the values in a single event attribute. In practice, there are also some limitations: while in \cref{def:complexPattern} we give no limit on the size of the behaviour matching set (i.e. the number of matching functions), the SIGNAL guide -- as of the writing of this paper -- states that only \emph{eight} such functions are allowed at once. We consider this and other such technical deviations to be of little importance to the effort of formalising SIGNAL. 

To define how simple or behavioural patterns  are evaluated, we introduce the concept of a segment, which is an interval inside an event set of a case. 

\begin{definition}[Segment]\label{def:segment}
Given an event log $L$, a case $c \in \mathit{Cases}(L)$ and its event set $E_c  \subseteq L$, a \emph{segment} $s \subseteq E_c$  is a subset of $E_c$ that contains for two time points $t_b, t_e \in \mathcal{T}$ all events $e \in E_c$ such that $t_b \leq \mathcal{T}(e) \leq t_e$, where we require that $\exists e', e'' \in s$ such that $\mathcal{T}(e') = t_b$ and $\mathcal{T}(e'') = t_e$. For simplicity, we can identify $s$ simply by $\langle t_b, t_e \rangle$, and we also introduce some needed notation: $b(s) = t_b$, $e(s) = t_e$.   
\end{definition}

Since segments are simply sets of events, we are free to use set operations on segments and still get segments as output\footnote{There are caveats, though: when using  set minus, an expression $s \setminus s'$,  can only produce a segment in the sense of \cref{def:segment}, if the time interval in $s'$ and $s$ satisfy the Allen's relations~\cite{DBLP:journals/cacm/Allen83} \emph{finishedBy} or \emph{startedBy} or $s = s'$. For set union, we require that the two segments involved have a non-empty intersection. Intersection has no special requirements.}\!. As a special case, we also introduce the empty set segment $s_{\emptyset}$, which will play a technical role in the definition below.

We can now define when a segment satisfies a given simple pattern. 

\begin{definition}[Simple Pattern Segment Satisfaction]\label{def:SimpleSegmentSatisfaction}
Given an event log $L$, a case $c \in \mathit{Cases}(L)$ and its event set $E_c \subseteq L$ and a simple pattern $\langle A_i, \mathcal{P}_s \rangle$ over $L$, we say that a segment $s$ satisfies $\langle A_i, \mathcal{P}_s \rangle$, if the following  holds: 
        \begin{enumerate}
        \item If $\mathcal{P}_s = a \in A_i$, then $\exists e \in s$ s.t. $A_i(e) = a$; 
        \item \label{point:second} if $\mathcal{P}_s = P' \lor Q$, then $s$ satisfies either $\langle A_i, P' \rangle$ or $\langle A_i, Q \rangle$;
        \item if $\mathcal{P}_s = \mathit{not}(P')$, then $s$ must not satisfy $\langle A_i, P' \rangle$; 
        \item if $\mathcal{P}_s = \mathit{start}(P')$, then s satisfies $\langle A_i, P' \rangle$ and $\not \exists e \in E_c $ with $\mathcal{T}(e) < b(s)$;
        \item if $\mathcal{P}_s = (P') \mathit{end}$, then s satisfies $\langle A_i, P' \rangle$ and $\not \exists e \in E_c $ with $\mathcal{T}(e) > e(s)$; 
        \item if $\mathcal{P}_s = \mathit{any}$, then $s$ trivially satisfies $\langle A_i, \mathcal{P}_s \rangle$; 
        \item \label{point:wiggly} if $\mathcal{P}_s = P' \rightsquigarrow Q $, then $\exists s', s'' \subseteq s$ with $e(s') < b(s'')$ and $s'$ satisfies $\langle A_i, P' \rangle$ and $s''$ satisfies $\langle A_i, Q \rangle$; 
        \item \label{point:succ} if $\mathcal{P}_s = P'  \rightarrow Q $, then  $s' \cup s'' = s$ in addition to all conditions from item \ref{point:wiggly};
        \item \label{point:last} if $\mathcal{P}_s = {P'}^*$, then either $\exists s' \subseteq s$ s.t. $s'$ satisfies $\langle A_i, P' \rangle $ and $s \setminus s'$ satisfies $\langle A_i, {P'}^* \rangle$, or we have $s = s_\emptyset$;
    \end{enumerate}  
    Note that in condition \ref{point:succ}, we require that $s'$ and $s''$ partition $s$. In other words, the segment satisfying $P'$ must be directly followed by the one satisfying $Q$. 
\end{definition}

\begin{figure}[t]

    \begin{minipage}{.5\linewidth}
      \captionof{table}{An example event set.}
      \label{tab:event_set}
      \centering
    \begin{tabular}{l | l } 
\toprule
\textbf{event\_name} & \textbf{timestamp} \\
\midrule 
$e_1 =$ \texttt{order\_received } &   $t_1 =$ 15.10.23 \\
$e_2 =$ \texttt{package\_collected } & $t_2 =$ 16.10.23 \\  
$e_3 =$ \texttt{package\_checked } &  $t_3 =$ 17.10.23 \\ 
$e_4 =$ \texttt{package\_sent } & $t_4 =$  23.10.23  \\ 
\bottomrule
\end{tabular}
    \end{minipage}%
    \begin{minipage}{.5\linewidth}
      \centering
        \captionof{table}{Simple pattern formulas.}
        \label{tab:patterns}
\begin{tabular}{l | l } 
\toprule
\textbf{pattern} & \textbf{Satis. segment} \\
\midrule 
 $e_2 \rightsquigarrow e_4$  &   $\langle t_2, t_4 \rangle$\\
$ \mathit{any} \rightsquigarrow e_4 $ & $\langle t_3, t_4 \rangle$\\  
$ \mathit{start}( e_1 \rightarrow e_2 ) $ & $\langle t_1, t_2 \rangle$\\ 
$ e_1 \rightarrow ( e_2  \rightsquigarrow e_4  )^ * $  & $\langle t_1, t_4 \rangle$\\ 
\bottomrule
\end{tabular} 
    \end{minipage}

\caption{ \cref{tab:event_set} contains an event set consisting of events with a single event attribute (\texttt{event\_name}) and a timestamp (shortened and formatted for brevity and human readability). \cref{tab:patterns} contains four simple pattern formulas and their satisfying segments over the event set of \cref{tab:event_set}.}
\label{fig:event_set}
\end{figure}

\begin{example}[Segment Examples] \label{ex:Segments}
In \cref{fig:event_set} we give an example of satisfying segments. For the sake of brevity, we only provide an event set, for some undefined  case, and thus omit the definition of a complete event log. 
\end{example}

We proceed to give the analogous definition for behavioural patterns. We shall reuse most of the conditions from \cref{def:SimpleSegmentSatisfaction}, as the structure of behavioural pattern formulas is mostly the same as for simple pattern formulas, except for the use of behavioural matching functions.

\begin{definition}[Behavioural Pattern Segment Satisfaction]\label{def:BehaviourSegmentSatisfaction}
    Given an event log $L$, a case $c \in \mathit{Cases}(L)$ and its event set $E_c  \subseteq L$ and a behavioural pattern $\langle B, \mathcal{P}_b \rangle$ over $L$, we say that a segment $s$ satisfies $\langle B, \mathcal{P}_b \rangle$, if the following conditions hold: 
        \begin{enumerate}
        \item If $\mathcal{P}_b = x$ where $\sigma_x \in B$, then $\exists e \in s$ s.t. $ \sigma_x(\mathit{att}(e)) = \top$; 
    \end{enumerate}
        in addition to conditions \ref{point:second} to \ref{point:last} from \cref{def:SimpleSegmentSatisfaction}.
\end{definition}

With this machinery, we can now clearly define how an \sigCC{} query is to be evaluated. For the parts of \sigCC{} that correspond to regular SQL, we use the familiar relational semantics. This only leaves the issue of how to deal with patterns. Informally speaking, patterns serve to filter out certain cases from our event log, where patterns are evaluated against the event set of each case individually. If no satisfying segment according to \cref{def:SimpleSegmentSatisfaction} (or resp. \cref{def:BehaviourSegmentSatisfaction}) exists, then all events of this case are to be removed. 

Our plan is to translate \sigCC{} to an extended form of RA, by introducing a new selection operator that takes as input an event log, and is given a (simple or behavioural) pattern as a parameter. Analogous to normal selection, it will filter out parts of the input table, namely those events in our event log that are from cases that do not have a satisfying segment for the pattern. 

\begin{definition}[Pattern Selection Operator]\label{def:PatternSelection}
    Given a simple  or behavioural pattern $\mathcal{P}$, and an event log $L$, we  define a \emph{pattern selection operator} $\sigma_{\mathcal{P}}$ as follows: 
    \[
        \sigma_{\mathcal{P}}(L) = \{  e \mid e \in L  \land E_{\mathit{Cid}(e)} \subseteq L \text{ has satisfying segment for } \mathcal{P}      \}
    \]
\end{definition}

We shall call the extension of RA with this new operator simply \RAExtended.  Formally, from the definitions for patterns and segment satisfaction, combined with \cref{def:PatternSelection}, we get the following corollary. 

\begin{proposition} \label{col:raExt}
Given an \sigCC{} query $q$, there exists an \RAExtended expression $\varphi$, s.t. for every event log $L$ it holds that:
\[
     \varphi(L) \equiv q(L).
 \] 
\end{proposition}

Note that in \cref{col:raExt} above,  we present both the query and the relational algebra expression as a function from an event log to a set of events. 

\bigskip

The translation of \sigCC{} to \RAExtended follows almost exactly the translation of CQs to RA -- recall that we selected \sigCC{} as a subset of SIGNAL that structurally matches CQs, with the sole addition of \texttt{\textbf{MATCHES}} clauses with pattern formulas. As we consider the translation of CQs to RA to be folklore, we omit technical details, and simply give an example below. 

\begin{example}[\sigCC{} to \RAExtended] \label{ex:RAE}
Consider the following \sigCC{} query from \cref{ex:SIGCC}:

        \begin{lstlisting}[language=SQL, 
        commentstyle=\itshape,
        basicstyle=\ttfamily,escapeinside={(*}{*)},morekeywords={MATCHES} ]
    SELECT case_id, event_name, event_time
      FROM eventlog 
     WHERE event_name MATCHES ('package_sent' (*$\rightsquigarrow$*) 'package_accepted')
        \end{lstlisting}

The query corresponds to the following \RAExtended{} expression:
\[
     \pi_{\text{\texttt{case\_id},\text{\texttt{event\_name}},\text{\texttt{event\_time}}}}\big ( \sigma_{\mathcal{P}}(\text{\texttt{eventlog})} \big )
\]
where $\mathcal{P} =\langle \text{\texttt{event\_name}}, \text{\texttt{package\_sent}} \rightsquigarrow \text{\texttt{package\_accepted} }  \rangle $.

\end{example}
\rev{
\paragraph{SIGNAL Queries with Nesting and Aggregation.}
In this section, we focused deliberately on a subset of \sig{}
which already captures its abilities to express complex patterns over event sets. 
We limit our study to such simple queries and deliberately leave out more expressive fragments, such as SIGNAL queries with nesting or the ability to aggregate over groups of attributes. The primary reason for this is that we believe that the ability to express patterns over event logs is already enough to set SIGNAL apart from standard relational algebra, as we shall show in the rest of this paper. Exploring how nesting further adds expressive power and affects data complexity is left for future work.
}

\section{Towards Lower and Upper Bounds on the Data Complexity of \emph{SCC} Evaluation}
\label{sec:bounds}

After having presented the semantics of \sigCC{} in \cref{sec:the_core_of_the_process_mining_query_language_sig}, we next consider how it fits into the complexity landscape of existing query languages. We will first show that \sigCC{} is indeed more expressive than relational algebra, and cannot be captured by it without the use of extensions, as were introduced in \cref{sec:the_core_of_the_process_mining_query_language_sig}. Next, we provide some preliminary work towards an exact characterisation of the evaluation problem by stating an upper bound via a reduction from \sigCC{} to semi-positive Datalog.

\subsection{Inexpressibility of \emph{SCC} in Relational Algebra}

In \cref{sec:the_core_of_the_process_mining_query_language_sig} we provide a translation of \sigCC{} to an extension of RA. It remains to show that this extension is actually necessary, i.e., that it is beyond the expressive capabilities of regular RA to capture the semantics of \sigCC{}.
We can achieve this by largely relying on well-known ``text-book''-level results.

\begin{theorem}[Inexpressibility of \sigCC{} in RA]\label{thm:inex}
    There exists an \sigCC{} query $q$, such that it is not possible to find an RA expression $\varphi$ for which we have that: 
    \[
        q(L) \equiv \varphi(L)
    \] for any given event log $L$.
\end{theorem}

\begin{proof}[Proof of \cref{thm:inex}]
Our argument will make use of Codd's Theorem~\cite{DBLP:persons/Codd72}, which states that the expressive power of relational algebra is the same as the domain-independent domain calculus. This in turn is equivalent in expressive power to first-order logic (FO). Hence, all we need to show is that the semantics of \sigCC{} cannot be captured in FO. To this end, we make use of the fact that FO cannot recognise \emph{evenness} in its models, i.e., models which only have an even number of elements. In the context of the relational model, this would correspond to the number of tuples in a relation.  For an easy to follow proof of this inexpressibility, we refer interested readers to~\cite{DBLP:books/sp/Libkin04}, where the desired result is captured by Proposition 3.3.
What is left to show is that there exists some \sigCC{} query that can express evenness. Specifically, this means that it will distinguish between event logs of even size and those that have an uneven size. We will make use of simple pattern matching for this, but we first need to establish that in \sigCC{} we can concatenate the events of all cases into one, single event set. This is possible by simply projecting out the \texttt{case\_id}, and replacing it with a single ID across the event log.
 This gives us a new event log, with the same number of events but where all events share the same \texttt{case\_id}. There are some technical complexities here, such as making sure the timestamps are still unique.
 We can accomplish this by also replacing the timestamps with stand-holder values, using for example the row number to indicate number of seconds. Afterwards, we filter this new event log by a simple pattern formula : $\mathit{start}(\mathit{any} \rightarrow \mathit{any})^*\mathit{end}$. We now have two possible outputs: either the event set (which corresponds to all events in this event log) has even size and thus has a satisfying segment: in this case we return every event in our event log. Alternatively, if the event log is of uneven size and thus permits no satisfying segment, we return nothing. Note the use of the $\mathit{start}$ and $\mathit{end}$ constructors, which are critical to capture this property: they ensure that the satisfying segment must span the entire event set.  \qed
\end{proof}

As a minor comment to the proof above, it should be noted that most implementations of SQL can effectively capture evenness via basic arithmetic. However, from a complexity standpoint, SQL with arithmetic is trivially undecidable, just as FO extended with (Peano) arithmetic is undecidable, so equating SQL with RA is sensible when talking about computational complexity.

\paragraph{On the role of the Kleene star in pattern formulas.}  To give some more insights into which components of a pattern formula are responsible for its expressive power, let us consider a subset of simple pattern formulas that do not involve the ${}^*$ constructor. In this case, patterns could not express repetitions, and every pattern formula would simply correspond to some fixed ordering over the event set. While we do have disjunction and negation, they are limited to event identifiers; more complex pattern constructors, such as $\rightarrow$, do not appear inside disjunctions. Thus, if we consider a pattern formula without the Kleene star, we could rewrite it into an FOL formula which captures the fixed ordering expressed in it, and so this fragment would stay within the expressive power of RA. Thus, it is not a coincidence that the counter-example we use in the proof above makes use if it, it is the crucial component that allows us to escape capture by FOL.
Let us also highlight that the Kleene star is apparently central to well-established approaches to reasoning about event log data. In particular DECLARE, a declarative language for process mining frequently used and studied by academics, features the Kleene star in all of its common expression templates when expressed in POSIX regular expressions~\cite{DBLP:conf/caise/CiccioBCM15} and many of the most central DECLARE expressions feature the Kleene star when translated to SIGNAL~\cite{DBLP:conf/bpm/BergmannRK23}.

Let us note that in contrast, the \verb|any| constructor is \emph{not} at the heart of the problem.
Intuitively, the example we use in the proof --  $\mathit{start}(\mathit{any} \rightarrow \mathit{any})^*\mathit{end}$ -- is equivalent to the regular expression \verb|^(..)*$|, and hence can be verified with a run-off-the-mill regular expression solver.
Analogously, $\mathit{start}(\mathit{`a`} \rightarrow \mathit{`a`})^*\mathit{end}$ (\verb|^(aa)*$| as a standard regular expression) still provides the needed counter-example, albeit not for \emph{any} event log, but for a specific event containing only events with name ``a''; for SIGNAL, the example can be trivially generalised to behaviours that match every event in an event log.


\subsection{Translation of \emph{SCC} to Semi-Positive Datalog } \label{sec:translation_of_sigcc_}
\vspace{-2mm}
Following the result that RA itself cannot capture every \sigCC{} query, we next want to find an existing formalism that does capture the semantics of \sigCC{} and thus get an upper bound on the complexity of evaluating \sigCC{}.
We proceed to show the following result.

\begin{theorem}[Reduction of \sigCC{} to Semi-Positive Datalog] \label{thm:RedDL}
    Given a \sigCC query $q$, and an event log $L$, we can define a semi-positive Datalog program $P$ \rev{with a linear ordering over timestamps}, such that: 
    \[ 
        q(L) \equiv P(L),
    \]
  where we understand $P(L)$ to be the output (generated tuples) of the Datalog program $P$ when evaluated with facts from the event log $L$. 
\end{theorem} 
\begin{proof}[Proof of Theorem \ref{thm:RedDL}] 
To simplify the presentation of the Datalog program, we make use of the functional dependencies of event logs, and identify each event with just the pair of its event id and case id, $\mathit{Eid}, \mathit{Cid}$, in addition to the timestamp as it plays a crucial role in capturing the semantics of patterns. We shall make use of the following extensional predicates. The first  is $\mathit{event}(c,e, t)$, where the three arguments are, respectively, the case id, the event id, and the timestamp. \rev{We encode with $\mathit{start}(c,t)$ (resp. $\mathit{end}(c,t)$) the timestamp of the chronologically first (resp. last) event in the event log of the case with id $c$.} We also have the ability to refer to individual attribute values via predicates of the form $A_i(c,e,a)$, which indicate that an event with case id $c$ and event id $e$ has value $a$ for attribute $A_i$. \rev{To ensure safety of any rules with negation, we require predicates $\mathit{time}$ and $\mathit{events}$ which identify, resp., all timestamps and event ids.}

\newcommand{\translate}[1]{[#1]^D}

\newcommand{\dl}[1]{ P_{#1}}

\newcommand{\seg}[1]{\mathit{segment}(#1)}

\newcommand{\eve}[1]{\mathit{event}(#1)}
\newcommand{\eves}[1]{\mathit{events}(#1)}
\newcommand{\timess}[1]{\mathit{time}(#1)}
\newcommand{\attr}[2]{\mathit{A}_{#1}(#2)}
\newcommand{\start}[1]{\mathit{start}(#1)}
\newcommand{\ed}[1]{\mathit{end}(#1)}


In order to capture when any given simple pattern has a satisfying segment, we use an intensional predicate  $\dl{\varphi}(t_s,t_e,c)$, which captures the satisfying segments $(t_s,t_e)$ of a given pattern formula $\varphi$ inside the events set $E_c$. We use the convention of  using  upper case for variables and lower case for constants. 
Here are the rules for intensional predicates $\dl{\varphi}$:
\addtocounter{footnote}{1}
\footnotetext[3]{Note that negation in simple pattern formulas can only occur for event identifiers so that in the translation, we need not consider the negation of more complex patterns.}
{ \small \begin{flalign*}
    \dl{a }(T,T,C) &\leftarrow \seg{T,T,C}, \eve{C,E,T},  \attr{i}{C,E,a} \\
    \dl{\neg a}(T,T,C) &\leftarrow \seg{T,T,C}, \eve{C,E,T},  \neg \attr{i}{C,E,a} \\
    \dl{ A \lor B }(T_s,T_e,C) &\leftarrow  \dl{ A  }(T_s,T_e,C) \lor \dl{B}(T_s,T_e,C) \\ 
     \dl{ \neg (A \lor B) }(T_s,T_e,C) &\leftarrow  \dl{ \neg A  }(T_s,T_e,C),  \dl{\neg B}(T_s,T_e,C) \\ 
    \dl{\neg \neg A}(T_s,T_e,C) &\leftarrow \dl{A}(T_s,T_e,C)  \text{ \itshape \text{\footnotemark[3]}} \\  
        \dl{\mathit{start}(A)}(T_s,T_e,C) &\leftarrow \dl{A}(T_s,T_e,C), \start{C,T_s} \\
    \dl{(A)\mathit{end}}(T_s,T_e,C) &\leftarrow \dl{A}(T_s,T_e,C), \ed{C,T_e}   \\
    \dl{A \rightsquigarrow B}(T_s,T'_e,C) &\leftarrow \seg{T_s, T'_e, C}, \dl{A}(T_s, T_e,C), \dl{B}(T'_s, T'_e,C), T_e < T'_s   \\ 
    \dl{A \rightarrow B}(T_s,T'_e,C) &\leftarrow \seg{T_s, T'_e, C}, \dl{A}(T_s, T_e,C), \dl{B}(T'_s, T'_e,C), T_e < T'_s   \\ 
 & \phantom{ = \text{ } } \neg  \eve{C,E,T''}, T_s < T'', T'' < T'_s, \timess{T''}, \eves{E} \\
\dl{A^*}(T_s,T'_e,C) &\leftarrow \seg{T_s,T'_e,C} \lor \big (    \dl{A}(T_s,T_e,C), \dl{A^*}(T'_s,T'_e, C), T_e < T'_s, \\
 & \phantom{ = \text{ } } \neg \eve{C,E,T''}, T_e < T'', T'' < T'_s, \timess{T''}, \eves{E} \big )
\end{flalign*}
In order for this translation to capture the intended meaning, we also require the following two rules to be included in the final Datalog program.
\begin{flalign*}
    \seg{T,T,C} &\leftarrow \eve{C,E,T} . \\ 
    \seg{T_s,T'_e,C} &\leftarrow  \seg{T_s, T_e, C}, \seg{T'_s, T'_e, C},  T_e < T'_s, \\ 
    & \phantom{ = \text{ }} \neg \eve{C, E , T''}, T_e < T'', T'' < T'_s , \eves{E},\timess{T''}
\end{flalign*}}
Note that in the translation above, we introduce a disjunction connective in the body of Datalog rules. This is simply ``syntactic sugar''.
{ \small\[
    \mathit{Atom}_1(A,B,C) \leftarrow \big (  \mathit{Atom}_2(A,B) \lor \mathit{Atom}_3(A,B) \big ), \mathit{Atom}_4(B,C) .
\]}
is a space-efficient way to encode the following: 
{\small\begin{flalign*}
            \mathit{Atom}_1(A,B,C) \leftarrow \mathit{Atom}_2(A,B), \mathit{Atom}_4(B,C) . \\
        \mathit{Atom}_1(A,B,C) \leftarrow \mathit{Atom}_3(A,B), \mathit{Atom}_4(B,C) . 
\end{flalign*}}
Equipped with these tools, we shall now present the translation of \sigCC{} into Datalog.
Let us first consider the case of \sigCC queries without any pattern matching. We can then model the query as a conjunctive query and the translation to Datalog is trivial: as we only have a single table---the event log---and can only join its event attributes or set them to constants, we can express the query inside a single Datalog rule: 
{ \small \[
    \mathit{Output}(X_1, \dots X_n) \leftarrow  \eve{C,E,T}, A_1(E,V_1), \dots A_n(E,V_\ell), V_1 = c_1, \dots, V_m = c_m .
\]}where $\bigcup_i X_i \subseteq \{ C,E,T, V_0, \dots, V_\ell  \}$ are the output variables, as determined in the SELECT clause. The joining of two event attributes is done by simply using the same variable symbol and when specifying the attribute via $A_i$. 

As such, we need only consider how to deal with pattern matching. For this, we  make use of the intensional predicate $\dl{}$:  for any given pattern formula $\varphi$, we add $\dl{\varphi}$ to the body of the rule above.  For every subexpression $\phi$ that occurs inside $\varphi$, we also add the rule $P_\phi$ to our Datalog program.

Given \cref{def:SimpleSegmentSatisfaction} and the construction above, it should not be difficult to see that our translation of simple patterns captures the intended meaning exactly. Each rule for a given pattern is intended to specify all the necessary conditions for the satisfying segment of that pattern to be fulfilled. We leave out the case for behaviour patterns, as it is analogous, safe that instead of event identifiers, we have  behaviour identifiers. 

What is left is to show is that our produced Datalog program after applying this translation, falls into the class of semi-positive Datalog. We note that negation can only occur in the translation of  patterns and in this translation we only negate the EDB predicates $\mathit{event}$  and $\mathit{A}_i$, which come from the event log.  \qed
\end{proof}

 The reduction to semi-positive Datalog to capture \sigCC{} gives us an upper bound on the data complexity of evaluating \sigCC{} queries, namely P.
This would directly lead to a completeness result, if one could show that the lower bound for \sigCC{} evaluation is also in P, in other words, that the bound is tight. 
\paragraph{Open Question.} \emph{Is the data complexity of \sigCC{} evaluation P-hard?}
\vspace{-2mm}

\section{Conclusion \& Future Work}
\label{sec:conclusion}
We formally presented and analysed a subset of the  SIGNAL query language, which forms the heart of a commercial process mining software product.
We provided the intended semantics for the pattern matching formulas that are used in SIGNAL to then show that SIGNAL cannot be captured in relational algebra.
Also, we conducted a preliminary analysis of the data complexity of the problem of evaluating \sigCC{} queries, establishing an upper bound of P-hard.
The findings highlight that there is more to process querying methods than ``just relational algebra'' and indicate that the temporal reasoning capabilities, which increase expressive power and raise the complexity bound, may indeed be good motivators of domain-specific implementations. Future work could answer the question posed in this paper regarding the tight upper bound of the SIGNAL core and conduct analogous analyses of other process querying methods such as PQL~\cite{DBLP:books/sp/22/Polyvyanyy22a}. For future work, we also leave the formal analysis of new process mining approaches such as object-centric process mining~\cite{DBLP:conf/caise/LiMCA18}, graph-based process querying~\cite{DBLP:conf/icpm/Jalali20,DBLP:conf/er/KhayatbashiHJ23,DBLP:journals/jodsn/EsserF21,DBLP:journals/spe/BeheshtiBM18}, or 
the modern wide-column store and OLAP-style systems (such as Apache Arrow) that underlie SIGNAL.
Another promising future research direction is the integration of our results with expressive power analyses of ``academic'' declarative process querying approaches, such as behavioural profiles~\cite{10.1007/s00165-016-0372-4} and DECLARE~\cite{DBLP:conf/bpm/GiacomoMGMM14}.

\bibliographystyle{splncs04}
\bibliography{main}

\end{document}